\documentclass[sigconf]{acmart} 

\copyrightyear{2020} 
\acmYear{2020} 
\setcopyright{acmlicensed}\acmConference[FODS '20]{Proceedings of the 2020 ACM-IMS Foundations of Data Science Conference}{October 19--20, 2020}{Virtual Event, USA}
\acmBooktitle{Proceedings of the 2020 ACM-IMS Foundations of Data Science Conference (FODS '20), October 19--20, 2020, Virtual Event, USA}
\acmPrice{15.00}
\acmDOI{10.1145/3412815.3416892}
\acmISBN{978-1-4503-8103-1/20/10}

\usepackage{mathrsfs}         
\usepackage[normalem]{ulem}   

\makeatletter
\floatstyle{ruled}
\newfloat{algorithm}{tbp}{loa}
\providecommand{\algorithmname}{Algorithm}
\floatname{algorithm}{\protect\algorithmname}
\numberwithin{equation}{section}
\makeatother

\newcommand{\InvSet}{\mathcal{S}^*}
\newcommand{\DP}{M}
\newcommand{\DPx}{M\left(x\right)}

\newcommand{\DPCo}{M^{*}}
\newcommand{\DPost}{f}
\newcommand{\DPosti}{f_{L_{1}}}
\newcommand{\DPostii}{f_{L_{2}}}

\begin{document}
\fancyhead{} 

\title[Congenial Privacy]{Congenial Differential Privacy under Mandated Disclosure}

\author{Ruobin Gong}
\affiliation{%
  \institution{Rutgers University}
  \streetaddress{110 Frelinghuysen Road}
  \city{Piscataway}
  \state{New Jersey}
  \postcode{08854}}
\email{ruobin.gong@rutgers.edu}

\author{Xiao-Li Meng}
\affiliation{%
  \institution{Harvard University}
  \streetaddress{1 Oxford St}
  \city{Cambridge}
  \state{Massachusetts}
  \postcode{02138}}
\email{meng@stat.harvard.edu}

\renewcommand{\shortauthors}{Gong and Meng}

\begin{abstract}	

Differentially private data releases are often required to satisfy a set of external constraints that reflect the legal, ethical, and logical mandates to which the data curator is obligated. The enforcement of constraints, when treated as post-processing, adds an extra phase in the production of privatized data. It is well understood in the theory of multi-phase processing that congeniality, a form of procedural compatibility between phases, is a prerequisite for the end users to straightforwardly obtain statistically valid results. Congenial differential privacy is theoretically principled, which facilitates transparency and intelligibility of the mechanism that would otherwise be undermined by ad-hoc post-processing procedures. We advocate for the systematic integration of mandated disclosure into the design of the privacy mechanism via  standard probabilistic conditioning on the invariant margins. Conditioning automatically renders congeniality because any extra post-processing phase becomes unnecessary. We provide both initial theoretical guarantees and a Markov chain algorithm for our proposal. We also discuss intriguing theoretical issues that arise in comparing congenital differential privacy and optimization-based post-processing, as well as directions for further research. 
	
\end{abstract}

%




\maketitle

\section{Privacy as data processing}\label{sec:intro}

\subsection{A blurry yet essential picture}
	
The curation and dissemination of large-scale datasets benefits science and society by supplying factual knowledge to assist discoveries, policy decisions, and promote transparency of information. As more data become accessible to more entities, however, the unobstructed access to information collected from individuals poses the risk of infringing on their privacy. Differential privacy is a mathematical concept that quantifies the extent of disclosure of confidential information in a database. It enjoys several advantages over its previous counterparts in statistical disclosure limitation. Most important of all, especially to data analysts who wish to conduct statistical inference on privatized data releases, is the transparency of the algorithm. The mechanism through which the privatized data release is generated can be spelled out in full, with its statistical properties fully understood. This enables analysts to incorporate it as part of a model, hence  permitting the statistical validity of the resulting inference \citep{gong2020transparent}.

The protection of confidential data with differential privacy relies on the careful design of a probabilistic mechanism, one that can veil the microscopic identities of individual respondents while preserving the macroscopic aspects of the data with high fidelity. The probabilistic nature of the mechanism is necessary, and enables a tradeoff as such to be made \citep{dinur2003revealing}. Typically, a differentially private mechanism injects a random perturbation into an otherwise deterministic query to be applied to the confidential database. One can say, then, that differentially private releases are ``blurry'' versions of the confidential data, just the same way a skilled Impressionist painter captures the essence of a pond of waterlilies without sketching out the contour of every petal and leaf.

When randomness is involved, however, certain truthful aspects of the data is invariably compromised, no matter how well-designed the privacy algorithm may be. Imagine if a picture of waterlilies was commissioned, not by an art collector, but by a botanist whose sole purpose is to study the structural formation of the plant, such as the exact length and width of its petals. She would be terribly disappointed at the Impressionist painting, even if it was the work of Claude Monet himself! 

Circumstances in practice dictates that aspects of the data release may be deemed as unfit to be tampered with. These are usually key statistics reflecting the fundamental purpose of data collection, as required by law, policy, or other external constraints as put forth by the stakeholders. The data curator is mandated to disclose these statistics accurately at any expense, while at the same time shielding the remainder part of the data release with a veil of privacy. This poses a challenge to the design of the privacy mechanism subject to mandated disclosure. The central question is, how to integrate data privatization, an inherently random act, with the mandated disclosure of partial yet deterministic information, while maintaining both the logical consistency of the overall data release and the quality of the privacy protection.

\subsection{Congenial privacy}

In this work we conceptualize the privacy mechanism as one of the many phases of data processing. The concept of congeniality, or rather {\it uncongeniality} \citep{meng1994multiple}, is then relevant. The theory of uncongeniality was developed for investigating a seemingly paradoxical phenomenon discovered by researchers dealing with imputations for the U.S. Decennial Census and similar public-use data files. \citeauthor{fay1992inferences} \citep{fay1992inferences} and \citeauthor{kott1995paradox} \citep{kott1995paradox} found that one can have inconsistent variance estimation for multiple imputation inference  \citep{rubin2004multiple}, even when both the imputation model and analysis procedure are valid.  The issue turned out to be a mathematical incompatibility between the imputation model and the analysis procedure, even if neither is incompatible with the underlying model that generates the data.  In other words, there is no probabilistically coherent model that can simultaneously imply the imputation model and the analysis procedure, a situation termed uncongeniality by \citeauthor{meng1994multiple} \citep{meng1994multiple}. To make matters worse, the imputation model, such as adopted by the Census Bureau, is typically not disclosed to the analyst of the imputed data, or at least not fully (e.g., due to confidential information used to help better predict the missing values).  The lack of transparency makes it impossible for the analyst to correct for the uncongeniality, and worse, even to realize the problem.  

The framework of uncongeniality was later generalized to the multivariate setting  \citep{xie2017dissecting} and to general multi-phase processing \citep{blocker2013potential}, which covers the current application by the same overarching principles. Two properties are critical for good privacy practice: transparency and congeniality. When the protection of privacy must observe mandated constraints, our proposal is to use conditional distributions as derived from the original unconstrained differential privacy mechanism  conditioning on invariant margins. This approach achieves both automatically.  Transparency is automatic, because the conditional distribution is determined by the original unconstrained distribution and the invariant constraints, both fully disclosed by design.   Congeniality stipulates the use of a \textit{single coherent probabilistic model} to ensure both differential privacy and mandated disclosure. This requirement is automatically satisfied when we use the conditional distribution derived from the original differential privacy mechanism, restricted solely to obey the mandated disclosure. A third advantage is that our proposal does not need any additional choice of procedural ingredients, such as projection distance, which is required by optimization-based post-processing such as adopted by the Census TopDown algorithm \citep{abowd2019census}.   
 
There is, however, no free lunch. The first price we pay for congenial privacy is computational. Sampling from a distribution truncated on some space, as determined by the invariant margins, is generally not trivial, especially when the truncated region is of irregular shape.  The second price is that we may pay more privacy loss budget than necessary, since the budget designed for the original mechanism depends on the sensitivity of the query measure on the unconstrained space,  which is larger than that for the constrained space.  When deriving the appropriate class of conditional distributions for the constrained mechanism, the new privacy loss budget should ideally be calibrated directly according to the query behavior on the constrained space, as opposed to be inherited from the unconstrained mechanism, which ensures a likely overly conservative level of privacy protection for the entire space. When the analytical complexity and computational requirements for the two approaches of budget calibration are similar, we certainly recommend the former.
 
\subsection{The mechanism of differential privacy}

Let $x = \left(x_1, \ldots, x_n\right)$ denote a database consisting of $n$ individuals, and $\mathcal{X}$ the space on which it is defined. A query function $s: \mathcal{X}\to\mathbb{R}^{d}$ embodies the knowledge contained in the database that stakeholders -- scientists, policy makers and the general public -- would like to learn. What determines the value of $s\left(x\right)$ is of course $x$, or equivalently all its component $x_i$ values, corresponding to individual respondents included in the database. It is precisely these individuals records, or the $x_i$'s, that are the subject of privacy protection. How can the data curator say useful things about $s\left(x\right)$, while saying barely anything about each of the $x_i$'s? 

The mechanism that can instill privacy into the curator's release appeals to randomness. A random query function $\DP:\mathcal{X}\to\mathbb{R}^{d}$ is said to satisfy $\epsilon$-differential privacy  \citep{dwork2006calibrating}, if for all pairs of neighboring datasets $\left(x,x'\right)\in\mathcal{X}\times\mathcal{X}$, we have that
\begin{equation}\label{eq:dp}
    P\left(\DPx\in B\right) \le\exp\left(\epsilon\right)P\left(\DP\left(x'\right)\in B\right)
\end{equation}
for all Borel-measurable sets $B\in\mathscr{B}\left(\mathbb{R}^{d}\right)$. In this work, the term {\it neighboring datasets} means that $x$ and $x'$ differ by exactly one individual's record, i.e. for some $j=1,\ldots,n$, $x_{j}\neq x_{j}'$ but  for all $i\neq j$, $x_{i}=x_{i}'$. Write $\mathtt{d}\left(x,x'\right)=1$ if $x$ and $x'$ are neighbors. This concept of neighbor is employed in the definition of $\epsilon$-bounded differential privacy \citep{abowd2019census}, and is distinct from the original formulation which defined neighbors as a pair that differ from each other by the addition or deletion of a single record.   
There are many ways to design an $\epsilon$-differentially private algorithm $\DP$, among which the most widely known and implemented are the Laplace and the Double Geometric algorithms \cite{dwork2006calibrating,fioretto2019differential}, both are additive $\epsilon$-differentially private mechanisms.

\begin{definition}[Laplace mechanism]
Let $s:\mathcal{X}\to\mathbb{R}^{d}$ be a deterministic query
function. The Laplace mechanism is given by
\begin{equation}\label{eq:additive}
\DPx:=s\left(x\right)+ \left(U_1,\ldots, U_d\right),
\end{equation}
where $U_i$'s are independent zero-mean Laplacian random variables each with dispersion parameter $\epsilon^{-1}\nabla\left(s\right)$, and
\[
\nabla\left(s\right)=\sup_{\left(x,x'\right)\in \mathcal{X}\times\mathcal{X} }\left\{ \left\Vert s\left(x\right)-s\left(x'\right)\right\Vert :\mathtt{d}\left(x,x'\right)=1\right\} 
\]
is the global sensitivity of $s$. When the database consists of binary records in an unrestricted domain, and $s$ is the counting or histogram query, $\nabla\left(s\right)=1$.
\end{definition}

\begin{definition}[Double Geometric mechanism]\label{def:2geo}
A random query mechanism $M$ is called the Double Geometric mechanism if it has the same functional form as \eqref{eq:additive}, with $U_i$ random variables defined on the integers with probability mass function
\begin{equation}\label{eq:2geo}
p_i\left(u \mid \epsilon \right) = \frac{1-a}{1+a}\cdot a^{\left|u\right|},
\end{equation}
where the parameter $a=a\left(\epsilon,s\right)=\exp\left(-\epsilon/\nabla\left(s\right)\right)$.
\end{definition}

Note that the definition of either the Laplace or the Double Geometric mechanism  presents not just one, but a collection of mechanisms that can be written as $\left\{ M_\epsilon\right\}$, indexed by $\epsilon > 0$ the {\it privacy loss budget} allocated to the mechanism in question. When regarded as a sequence of statistical procedures, $\epsilon$ serves as an indicator of the statistical quality of the output (with larger $\epsilon$ for higher quality) that can be used to offer interpretation and to guide its own choice.
This point will be immediately useful in Section~\ref{subsec:intelligibility}.

\subsection{Statistical intelligibility }\label{subsec:intelligibility}

An important reason that differential privacy is embraced by the statistics community is that it defines privacy in the language of probability, induced by the mechanism that injects randomness in the data release. An impactful consequence is that the distributional specification of the mechanism can be made fully transparent without sabotaging the promised protection. This opens the door for systematic analysis of the statistical property of mechanisms, which is in turn crucial to the accurate interpretation of statistical inference from privatized data releases \cite{gong2019exact}. The clarity both in definition and in implementation makes up the statistical intelligibility of differential privacy as a data processing procedure.

We discuss the statistical interpretation of the privacy mechanism, which is what served as inspiration for the conditioning approach to construct the invariant-respecting mechanism in the first place. The degree of protection exerted by a privacy mechanism on the confidential database is seen as a calculated limit on the statistical knowledge it is able to supply, as a function of the privacy loss budget allotted to the mechanism. Compare quantities
\begin{equation}\label{eq:pair2}
\pi\left(x_{i}=\omega\right)\;\;\text{and}\;\;\pi\left(x_{i}=\omega\mid M_{\epsilon}\left(x\right)\in B\right),	
\end{equation}
which are the analyst's prior probability about the value of the $i$th entry of the dataset, versus her posterior probability if an $\epsilon$-differentially private query released a report $B$. 
Thus, the statistical meaning of $M_{\epsilon}$ can be explained as follows.

\begin{theorem}\label{thm:posterior}
Let $x = \left(\cdots, x_i,\cdots \right) \in \mathcal{X}$ be the database, and $\{ M_{\epsilon}: \mathcal{X} \to \mathbb{R}^{d}, \epsilon > 0 \}$ a class of $\epsilon$-differentially private procedures operating on $x$. Then for every $B\in\mathscr{B}\left(\mathbb{R}^{d}\right)$, $\epsilon > 0$ and every prior probability $\pi$ the analyst harbors about $x_{i}$, 
\begin{multline}\label{eq:bound}
\pi\left(x_{i}=\omega\mid M_{\epsilon}\left(x\right)\in B\right) \in \\ \bigg[\exp\left(-\epsilon\right)\pi\left(x_{i}=\omega \right),\  \exp\left(\epsilon\right)\pi\left(x_{i}=\omega \right)\bigg].
\end{multline} 
\end{theorem}

\begin{proof}
The posterior probability $\pi\left(x_{i}=\omega\mid M_{\epsilon}\left(x\right)\in B\right)$ can be written as
\begin{align*}
&\frac{P\left(M_{\epsilon}\left(x\right)\in B\mid x_{i}=\omega\right)\pi\left(x_{i}=\omega\right)}{\sum_{\omega'}P\left(M_{\epsilon}\left(x\right)\in B\mid x_{i}=\omega'\right)\pi\left(x_{i}=\omega'\right)}.
\end{align*}
The result then follows immediately from the fact that 
$M_{\epsilon}$ is $\epsilon$-private, which means that for any 
$B\in\mathscr{B}\left(\mathbb{R}^{d}\right)$, 
\[
\exp{(-\epsilon)}\le \frac{ P\left(M_{\epsilon}\left(x\right)\in B\mid x_{i}=\omega'\right)}{P\left(M_{\epsilon}\left(x\right)\in B\mid x_{i}=\omega\right)}\le  \exp{(\epsilon)}. 
\]
\end{proof}

Theorem~\ref{thm:posterior} says that, any release generated by an $\epsilon$-differentially private procedure sharpens the analyst's knowledge about $x_{i}$ by at most a factor of $\exp(\epsilon)$. This interpretation provides a direct link between the differential privacy promise and the actual posterior risk of disclosure due to the release of the random query $M_{\epsilon}$.

Recall that the definition of the Laplace and the Double Geometric mechanisms are both well-defined for any $\epsilon > 0$. However, in the limiting case of $\epsilon \to 0$,   i.e. the privacy loss budget becomes increasingly restrictive, both algorithms amount to adding noise with increasingly large variance to the confidential query. At $\epsilon = 0$, neither mechanism remains well-defined since the distributions of the noise component become improper due to the infinite cardinality of their respective domains. Nevertheless, the definition of $\epsilon$-differential privacy allows for the expression with $\epsilon=0$. A mechanism is $0$-differentially private if one cannot gain any discriminatory knowledge from its release about the underlying database whatsoever. In other words, the analyst's knowledge about the individual state of $x_i$ must remain the same as her prior. This notion can be explained consistently with Theorem~\ref{thm:posterior}, by observing that
\begin{equation}
\lim_{\epsilon\to0}\pi\left(x_{i}=\omega\mid M_{\epsilon}\left(x\right)\in B\right)=\pi\left(x_{i}=\omega\right),	
\end{equation}
where the limit is implied by \eqref{eq:bound}. 
This inspires the following deliberate construction of $M_{0}$ as a $0$-differentially private procedure.

\begin{definition}[$0$-differentially private procedure]\label{def:0dp}
For $\{ M_{\epsilon}\}$ a class of $\epsilon$-differentially private procedures well-defined for $\epsilon > 0$ but not $\epsilon = 0$, define $M_0$ as the $0$-differentially private procedure such that for every prior probability $\pi$,
\begin{equation}\label{eq:s0x}
\pi\left(x_{i}=\omega\mid M_{0}\left(x\right)\in B\right)=\pi\left(x_{i}=\omega\right),\; \forall B\in\mathscr{B}\left(\mathbb{R}^{d}\right).	
\end{equation}
\end{definition}

Definition~\ref{def:0dp} grants conceptual continuity to $M_{0}$, a perfectly meaningful object in the privacy sense but lacking statistical intelligibility from the mechanistic point of view. For practical purposes, $M_{0}$ should be taken to mean the {\it suppression} procedure, which supplies {\it vacuous} knowledge to the analyst about the state of affairs of the database. Theoretically, the meaning of  $M_{0}$ as a probabilistic mechanism cannot be supplied by ordinary probabilities, because any probability specification represents a set of specific knowledge about relative frequencies of any pair of states.
However, its meaning can be quantified precisely in the more general framework of \textit{imprecise probability}, as Section~\ref{sec:discussion} will discuss.

\section{Constructing congenial privacy}\label{sec:coprocess}

\subsection{Privacy with invariants}

While privacy protection is called for, the data curator may be simultaneously mandated to disclose certain aspects of the data as they are exactly observed, without subjecting them to any privacy protection. This collection of information is referred to as {\it invariant information}, or {\it invariants}. 
In practice, invariants are often defined according to a set of exact statistics calculated based on the confidential database \citep{ashmead2019effective}. For example, suppose $s$ is the histogram query which tabulates the population residing in each county of the state of New Jersey from the 2020 U.S. Census. When producing a differentially private version of the histogram, the Census Bureau is constitutionally mandated to report the total population of each state as enumerated.  This means that the privatized histogram $\DP (x^*)$ must possess the same total population size as $s(x^*)$, where $x^*$ the confidential Census microdata; or in notation, $\left\Vert \DP\left(x^*\right) \right\Vert = \left\Vert s\left(x^{*}\right)\right\Vert$.

Suppose that the Double Geometric mechanism is to be applied to the histogram query $s$. Due to the random nature of the perturbations, a single realization of the mechanism will with high probability produce $\left\Vert \DP\left(x^*\right) \right\Vert \neq \left\Vert s\left(x^{*}\right)\right\Vert$.
Furthermore $\DP\left(x^*\right)$ has a positive probability of consisting negative cell counts, which is logically impossible of the confidential query $ s\left(x^{*}\right) $.  The challenge of privacy preservation under mandated disclosure is thus to find an alternative mechanism, 
say $\tilde{M}$, such that every realization of $\tilde{M}$ meets all the requirement of mandated information disclosure, 
while preserving the promise of differential privacy.

Let $\mathcal{X}^{*}\subset \mathcal{X}$ be the set of $x$'s that obey the given invariants.
In turn,  $\mathcal{X}^{*}$ defines the set of values that the query must satisfy as
\begin{equation}\label{eq:invset}
\InvSet=\left\{ s\left(x\right)\in\mathbb{R}^{d}:x\in\mathcal{X}^{*}\right\}.	
\end{equation}
Note that implicitly, $\InvSet = \InvSet\left(x^{*}\right)$ is a set-valued function of the confidential dataset $x^{*}$, because the invariant knowledge we intend to impose on the private release is implied by $x^{*}$. 

A random mechanism $\tilde{M}$ satisfies the mandated disclosure if
\begin{equation}
\tilde{M}\left(x\right)\in\InvSet,\quad\forall x\in\mathcal{X}^{*}.
\end{equation}
That is, whenever applied to a database conformal to the mandated disclosure,  with mathematical certainty $\tilde{M}$ is also conformal to the mandated disclosure. The size and complexity of the restricted $\InvSet$ (and $\mathcal{X}^{*}$) relative to their original spaces are crucial to the overall extent to which privacy of the residual information in the database can be protected, a point we will revisit in Section~\ref{subset:gamma}. 

There may be many ways to construct a random mechanism $\tilde{M}$, but all are not equally desirable. We argue that $\tilde{M}$ should be constructed in a principled manner, and a constructive way to achieve that is to use conditional distributions of unconstrained privacy mechanisms. The resulting mechanism can be easily tuned to
retain its differential privacy promise, while ensuring its congenial integration into the data processing pipeline, preserving the statistical intelligibility of its releases.

\subsection{Imposing invariants via conditioning}\label{subsec:cond}

Let $\DP$ be a valid $\epsilon$-differentially private mechanism, which generates outputs that typically do not obey the invariant requirement $\DPx \in \InvSet$, even if $x \in \mathcal{X}^{*}$. A natural idea to force the requirement onto $\DP$ is via conditioning. Define a modified privatization mechanism $\DPCo$, such that the probability distribution it induces is the same as the conditional distribution of $\DP$ subject to the constraint that $ \DPx\in\InvSet$. For what's next, we'll use the notation $Z\overset{L}{=}W$ to mean $Z$ and $W$ are identically distributed, and $\DPx\mid \DPx\in\InvSet$ denotes a well-specified conditional distribution $P(\DPx\mid \DPx\in\InvSet)$. Also assume for now $P(\DPx\in\InvSet)>0$. We have the following theorem.

\begin{theorem}\label{thm:2epsilon}
 Let $x^{*} \in \mathcal{X}$ be the confidential database, and  $\mathcal{X}^{*} \subset \mathcal{X}$ the invariant subset to which $x^*$ conforms. The deterministic function $s:\mathcal{X}\to\mathbb{R}^{d}$ is a query, and the implied $\InvSet\in\mathscr{B}\left(\mathbb{R}^{d}\right)$ is defined by \eqref{eq:invset}. Let $M$ be an $\epsilon$-differentially private mechanism based on $s$, and $\DPCo$ be a constrained mechanism such that
\begin{equation}
\DPCo\left(x\right)\overset{L}{=}\DPx\mid \DPx\in\InvSet.
\end{equation}
Then for all invariant-conforming pairs of datasets that are $k$-neighbors, i.e. $\left(x,x'\right) \in \mathcal{X}^{*}\times \mathcal{X}^{*}$ such that $\mathtt{d}\left(x,x'\right)=k$, there exists a real-valued $\gamma \in [-1, 1]$ such that for all $B\in\mathscr{B}\left(\mathbb{R}^{d}\right)$, 
\begin{equation}\label{eq:main-theorem}
P\left(\DPCo\left(x\right)\in B\right)\le \exp\left(\left(1 + \gamma\right)k\epsilon\right)P\left(\DPCo\left(x'\right)\in B\right).
\end{equation}
\end{theorem}

\begin{proof}
For a pair of $k$-neighboring and $\InvSet$-conforming datasets $\left(x,x'\right)$ and any $B\in\mathscr{B}\left(\mathbb{R}^{d}\right)$,
\begin{align*}
\frac{P\left(\DPCo\left(x\right)\in B\right)}{P\left(\DPCo\left(x'\right)\in B\right)} & =  \frac{P\left(\DPx\in B\mid \DPx\in\InvSet\right)}{P\left(M\left(x'\right)\in B\mid M\left(x'\right)\in\InvSet\right)} \\ 
 & =  \frac{P\left(\DPx\in B\cap\InvSet\right)}{P\left(M\left(x'\right)\in B\cap\InvSet\right)}  \cdot \frac{P\left(M\left(x'\right)\in\InvSet\right)}{P\left(\DPx\in\InvSet\right)}.
\end{align*}	
Clearly each of the last two ratios above is bounded above by $\exp\left(k\epsilon\right)$ and below by
$\exp\left(-k\epsilon\right)$ because $M$ is $\epsilon$-differentially private. Consequently, if we let
\begin{align}
\gamma^* =\frac{1}{{k\epsilon}} \log  \left[\max_{\substack{\left(x,x'\right)\in\mathcal{X}^*\times\mathcal{X}^*\\
\mathtt{d}\left(x,x'\right)=k
}
} \frac{P\left(M\left(x'\right)\in\InvSet\right)}{P\left(\DPx\in\InvSet\right)}\right], \label{eq:gamma_star}
\end{align}
then $\gamma^*\in [0, 1]$ and it is a known constant because $\InvSet$ is public. Then, \eqref{eq:main-theorem} holds for any $\gamma \in \left[\gamma^*, 1 \right]$, and certainly for $\gamma = 1$. 
\end{proof}

Our proof above might create an impression that only $\gamma\ge 0$ is permissible, but Sections~\ref{sec:postprocess} and~\ref{subset:gamma} will supply two examples both with $\gamma < 0$.
Negative $\gamma$ may sound paradoxical, for it seems to suggest that better privacy protection can be achieved by disclosing some information. However, we must be mindful that differential privacy is not about protecting privacy in absolute terms.  Rather, it is about controlling the \textit{additional} disclosure risk from releasing the privatized data to the users (or hackers), relative to what they know before the release. The presence or absence of mandated invariants would \textit{ex ante} constitute two different bodies of knowledge, hence any additional privacy protection would carry different interpretations too.
We also emphasize that Theorem~\ref{thm:2epsilon} generalizes to cases where $P(M(x) \in \InvSet)=0$, such as when it is a linear subspace of $\mathbb{R}^{d}$ and the privacy mechanism is continuous. The proof is a bit more involved in order to properly define $P(\DPx \mid \DPx\in \InvSet)$, which we will discuss in future work. However, this complication is not a concern for discrete privatization mechanisms, such as within the Census TopDown algorithm \citep{abowd2019census}. 

Theorem~\ref{thm:2epsilon} holds broadly for arbitrary kinds of $\epsilon$-differentially private mechanisms, as well as any deterministic invariant information about either the database or the query function that can be expressed in a set form. It lends itself to the same kind of posterior interpretation enjoyed by unconstrained differentially private mechanisms. Specifically, if $\DPCo_{\epsilon}$ is the constrained differentially private procedure constructed based on the unconstrained procedure $M_{\epsilon}$, according to the specifications of Theorem~\ref{thm:2epsilon}, then for all $x\in\mathcal{X}^{*}$ such that $\exists  x'\in\mathcal{X}^{*}$ so that $\mathtt{d}\left(x,x'\right)=1$, and  $\forall B\in \mathscr{B}\left(\InvSet\right)$, the analyst's posterior probability $\pi\left(x_{i}=\omega\mid x\in\mathcal{X}^{*},M_{\epsilon}^{*}\left(x\right) \in B\right)$ is bounded in between 
\begin{multline*}
\bigg[\exp\left(-\left(1+\gamma\right)\epsilon\right)\pi\left(x_{i}=\omega\mid x\in\mathcal{X}^{*}\right), \\ \exp\left(\left(1+\gamma\right)\epsilon\right)\pi\left(x_{i}=\omega\mid x\in\mathcal{X}^{*}\right)\bigg].
\end{multline*} 
This an interval that bears structural resemblance to~\eqref{eq:bound}, thanks to the conditional nature of $M^{*}$ which allows for statistical information from privacy mechanisms, constrained or otherwise, to be interpreted in the same (hence congenial) way. 

The definition of $\epsilon$-differential privacy has the property that, if a mechanism $M$ is $\epsilon_1$-differentially private, then for all $\epsilon_2 \ge \epsilon_1$, $M$ is also $\epsilon_2$-differentially private. If the invariant information does not {\it substantially disrupt} the neighboring structure of the sample space of the database, a notion we will make precise later in Section~\ref{sec:discussion}, what Theorem~\ref{thm:2epsilon} says is that enforcing the invariant $\mathcal{X}^{*}$ onto the unconstrained mechanism $M$ via conditioning costs $\left(1+\gamma\right)$ times -- and {\it at most twice} since $\gamma$ can always be set to $1$ -- the privacy loss budget allotted to $M$. When a more cost-effective value of $\gamma$ is hard to determine, a simplest way to ensure privacy loss budget $\epsilon$ for $\DPCo$ is to use a budget of $\epsilon_{0}=\epsilon/2$ for the unconstrained mechanism $M$ to begin with; see Section~\ref{sec:example}.

\subsection{An Monte Carlo Implementation}\label{sec:computation}

Let $p$ denote the probability distribution, either a mass function or a density function, induced by the unconstrained differentially private algorithm $\DP$ which depends on the confidential query $s^{*}$. Further denote $p^{*}$ to be the corresponding conditional distribution of $p$ constrained on the invariant set $\InvSet$. 
The constrained privacy mechanism requires samples from $p^{*}$. A simplest, though often inefficient or event impractical, method is rejection sampling. Since 
\[
p^{*}\left(s\right)=\begin{cases}
\begin{array}{c}
c^{-1}p\left(s\right)\\
0
\end{array} & \begin{array}{c}
\text{if }s\in\InvSet\\
\text{otherwise},
\end{array}\end{cases}
\]
where $c=\int_{\InvSet}p\left(s\right)ds$,
one can opt for a proposal density $q$ with support $\InvSet$ such that $\sup_{s\in \InvSet} [p(s)/q(s)]\le R$, and accept a sample $s\sim q$ with probability $p(s)/Rq(s)$. This encompasses the option to set $q=p$, the unconstrained
privacy kernel itself, and keep sampling until the sample falls into $\InvSet$. This strategy is clearly inefficient in general, and impossible when $c=0$. 

Efficient algorithm tailor-made to specifications of $\InvSet$ are possible. Here, we present in Algorithm~\ref{alg:MIS} a generic approach based on Metropolized Independent Sampling \citep[MIS;][]{liu1996metropolized}, for the most common case in which the mandated invariants are expressed in terms of a consistent system of linear equalities and inequalities
\begin{equation*}
\InvSet=\left\{ s \in \mathbb{R}^{d}:As=a,Bs\ge b\right\}.
\end{equation*}
Here $A$ and $B$ are $d_{A}\times d$ and $d_{B}\times d$ matrices with ranks $d_{A}< d$ and $d_{B}< d$ respectively, and $a$ and $b$ vectors with length $d_{A}$ and $d_{B}$ respectively. The algorithm requires a proposal index set $\mathcal{I}$, a subset of $\left\{ 1,\ldots,d\right\}$ of size $d-d_{A}$ such that $\text{rank}\left(A_{[\mathcal{I}^{c}]}\right)=d_{A}$, where $A_{[\mathcal{J}]}$ is the submatrix of $A$ consisting of all columns whose indices belong to $\mathcal{J}$.  For each $A$, the choice of $\mathcal{I}$ may not be unique, and can have a potential impact on the efficiency of the algorithm.

\begin{algorithm}[!h]
\caption{Metropolized Independent Differentially Private Sampler with Invariants}
\label{alg:MIS}
\begin{tabbing}
\enspace{}Input: unconstrained privacy mechanism $p$, \\ 
\enspace{}\;\, \qquad{}confidential query $s^*$, invariant parameters $(A,a,B,b)$, \\
\enspace{}\;\, \qquad{}proposal distribution $q$, proposal index set $\mathcal{I}$, \\
\enspace{}\;\, \qquad{}initial value $s^{(0)}\in\InvSet$, integer $\text{nsim}$; \\
\enspace{}Iterate: for $t=0,1,\ldots,\text{nsim}-1$, at $t+1$: \\
\qquad{}step 1, propose $\tilde{s}$: \\ 
\qquad{} \qquad{}\; 1-1. sample $\tilde{s}{}_{\mathcal{I}}\sim q$; \\
\qquad{} \qquad{}\; 1-2. solve for $\tilde{s}_{\mathcal{I}^{c}}$ in $A_{\left[\mathcal{I}\right]}\tilde{s}_{\mathcal{I}}+A_{\left[\mathcal{I}^{c}\right]}\tilde{s}_{\mathcal{I}^{c}}=a$; \\ 
\qquad{} \qquad{}\; 1-3. write $\tilde{s}=\left(\tilde{s}_{\mathcal{I}},\tilde{s}_{\mathcal{I}^{c}}\right)$; \\ 
\qquad{}step 2, compute $\alpha(s^{\left(t\right)},\tilde{s})=\min\left\{ 1,\frac{p(\tilde{s}){\bf 1}\left(B\tilde{s}\ge b\right)q\left(s_{\mathcal{I}}^{(t)}\right)}{p\left(s^{(t)}\right)q\left(\tilde{s}_{\mathcal{I}}\right)}\right\}$; \\
\qquad{}step 3, set $s^{\left(t+1\right)}=\tilde{s}$ with probability
$\alpha(s^{\left(t\right)},\tilde{s})$, \\
\qquad{} \qquad{}\; otherwise set
$s^{\left(t+1\right)}=s^{\left(t\right)}$. \\
\enspace{}Output: a set of draws $\{s^{\left(t\right)}\}$, $t = 1,\ldots,\text{nsim}$. 
\end{tabbing}	 
\end{algorithm} 

This algorithm is applicable to both discrete and continuous data and privatization schemes, and it does not require the normalizing constant for $p^{*}$.
In Section~\ref{sec:example} below, we use it to construct a mechanism for differentially private demographic contingency tables subject to both linear equality and inequality invariant constraints.

\section{Contingency Table with Invariants}\label{sec:example}

\begin{table*}[ht]
\centering
\begin{tabular}{rllllllllllllll}
  \toprule
 & $<$ 5 & 6-10 & 11-15 & 16-17 & 18-19 & 20 & 21 & 22-24 & 25-29 & 30-34 & 35-39 & 40-44 & 45-49 & 50-54 \\ 
  \midrule
Female & 8 & 6 & 3 & 6 & 4 & 4 & 4 & 8 & 5 & 7 & 7 & 6 & 1 & 5 \\ 
  Male & 3 & 4 & 5 & 8 & 6 & 4 & 5 & 5 & 5 & 6 & 10 & 7 & 3 & 2 \\ 
   &  &  &  &  & 55-59 & 60-61 & 62-64 & 65-66 & 67-69 & 70-74 & 75-79 & 80-84 & 85+ & Total \\   \cmidrule{6-15}
  Female &  &  &  &  & 4 & 4 & 9 & 6 & 2 & 8 & 8 & 8 & 7 & {\bf 130} \\ 
  Male &  &  &  &  & 5 & 11 & 6 & 4 & 7 & 4 & 5 & 3 & 8 & {\bf 126} \\ 
  \midrule
   Voting &  &  &  & {\bf 43} &  &  &  &  &  &  &  &  & {\bf 213} & {\bf 256} \\ 
\bottomrule
\end{tabular}

\begin{tabular}{rllllllllllllll}
  \toprule
 & $<$ 5 & 6-10 & 11-15 & 16-17 & 18-19 & 20 & 21 & 22-24 & 25-29 & 30-34 & 35-39 & 40-44 & 45-49 & 50-54 \\ 
  \midrule
Female & 6 & 6 & 4 & 3 & 2 & 10 & 2 & 5 & 6 & 7 & 5 & 6 & 0 & 5 \\ 
  Male & 9 & 4 & 5 & 6 & 3 & 4 & 5 & 5 & 3 & 4 & 7 & 8 & 5 & 3 \\   \cmidrule{6-15}
   &  &  &  &  & 55-59 & 60-61 & 62-64 & 65-66 & 67-69 & 70-74 & 75-79 & 80-84 & 85+ & Total \\  \cmidrule{6-15}
  Female &  &  &  &  & 4 & 5 & 10 & 8 & 3 & 8 & 8 & 12 & 5 & {\bf 130} \\ 
  Male &  &  &  &  & 2 & 23 & 8 & 2 & 6 & 3 & 0 & 3 & 8 &  {\bf 126} \\ 
  \midrule
   Voting &  &  &  & {\bf 43} &  &  &  &  &  &  &  &  & {\bf 213} & {\bf 256} \\  
   \bottomrule
\end{tabular}

\caption{A confidential sex $\times$ age contingency table (top) and a corresponding constrained differentially private (bottom) release, subject to total population, proportion female population and voting age population constraints (bold).  \label{tab:mis-sample}}
\end{table*}


The table we consider is of dimension $2 \times 23$, with rows representing sex (male/female), and columns representing age bucketed roughly every four years, with finer buckets around key age ranges such as 18, 21 and 60. This data structure corresponds to the 2010 Census Table \#P12 and 2020 Census Table \#P7 \cite{census20202010}, one of the most referenced type of contingency table releases by the Census Bureau at various geographic levels. For the purpose of computational illustration, this example will use simulation to construct synthetic datasets that represent the confidential Census demographic data.

 Let $s$ be a vector of length $46$, denoting the row-vectorized contingency table. The constraints to be imposed on the differentially private table include 
\begin{enumerate}
	\item total population;
	\item proportion of female population;
	\item total voting age population (18+ age); and
	\item nonnegative table entries.
\end{enumerate}
Items (1) to (3) constitute equality constraints, and item (4) inequality constraints. 
The unconstrained privacy mechanism that serves as the basis of our construction is the Double Geometric mechanism, with distribution function 
\begin{equation*}
p\left(s\right)=\prod_{i=1}^{46}p_{i}\left(s_{i}-s_{i}^{*}\mid\epsilon\right),
\end{equation*}
where $p_i$ is as defined in \eqref{eq:2geo}, and the privacy loss budget set to $\epsilon = 0.5$ per cell. The proposal distribution is set to be of the same family as does the unconstrained privatization algorithm, but it is given a distinct dispersion parameter $\tilde{\epsilon}$ to tune for best algorithmic performance, in this case the acceptance probability of the algorithm. The proposal distribution function for $\tilde{s}_{\mathcal{I}}$, the ($d-d_{A}$)-length subvector of the $k$th proposal $\tilde{s}$, is 
\begin{equation*}
q\left(s\right)=\prod_{i\in \mathcal{I}}p_{i}\left(s_{i}-s_{i}^{*}\mid\tilde{\epsilon}\right),
\end{equation*}
where $\mathcal{I}$ is chosen to be $\left\{2,\ldots,22, 24,\ldots 45\right\}$. The remainder coordinates of the $k$th proposal, $\tilde{s}_{\mathcal{I}^{c}}$, is solved according to the equality constraint $A\tilde{s}=a$. 

Table~\ref{tab:mis-sample} displays a simulated confidential table (top) and a constrained differentially private table (bottom) based on the confidential table, where the draw is produced by Algorithm~\ref{alg:MIS}. In this case, setting $\tilde{\epsilon} = 0.6$ yields the best acceptance probability, with acceptance probability at $1.68\%$, as shown in Figure~\ref{fig:mis-dispersion} in Appendix~\ref{appendix:mis} alongside traceplots for the second and the first cells of the table, with the former index belonging to $\mathcal{I}$ and the latter not.

 \begin{figure}[h]
   \centering
   \includegraphics[width=.45\textwidth]{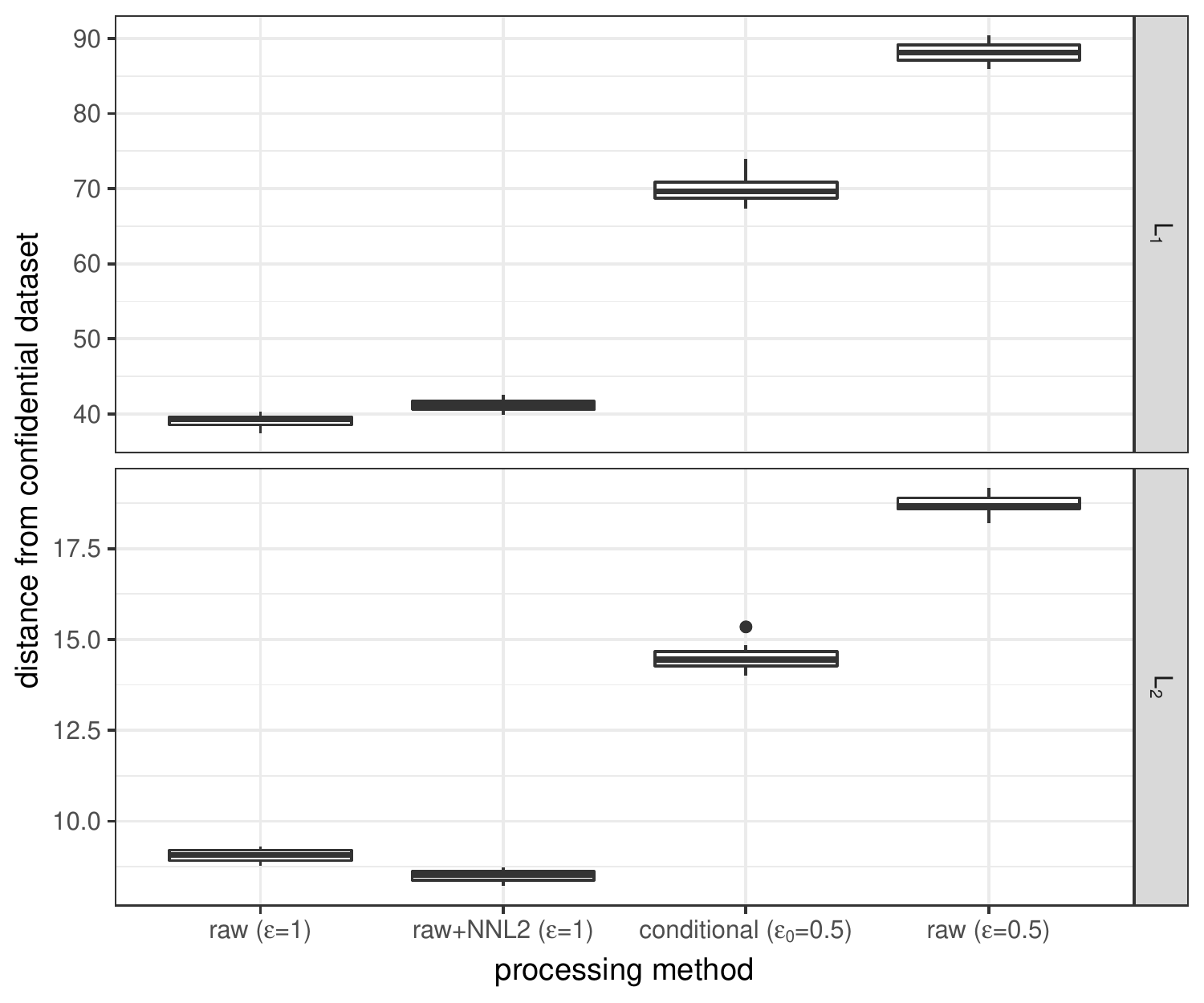}
   \caption{Average  $L_1$ (top) and $L_2$ (bottom) distances of a simulated confidential dataset from its privatized releases using four different processing methods. \label{fig:mis-variance}}
 \end{figure}
 
  We compare the outputs of our congenial mechanism with the unconstrained privacy mechanism, with and without post-processing via nonnegative $L_{2}$ minimization onto the invariant set $\InvSet$ defined by (1)-(4). A total of $20$ confidential contingency tables are simulated, each with cells following 
 \begin{equation*}
 s_{i}\overset{i.i.d.}{\sim}Negative\  Binomial\left(100,.05\right),	
 \end{equation*}
 which has mean $E(s_i) = 5.26$ and variance $Var(s_i) = 5.54$. For each confidential table, $100$ privatized releases were created using each of the following methods: 
\begin{enumerate}
\item[a)] the unconstrained (raw) Double Geometric mechanism with privacy loss budget $\epsilon = 1$ per cell; 
\item[b)] the above mechanism followed by nonnegative $L_{2}$ (NNL2) minimization onto the subspace defined by the four constraints, i.e. $\InvSet$; 
\item[c)] Our proposed $\InvSet$-conditional algorithm per Algorithm~\ref{alg:MIS}, constructed based on an unconstrained Double Geometric mechanism with $\epsilon_0 = 0.5$; and
\item[d)] the same unconstrained Double Geometric mechanism as in (a), but with $\epsilon = 0.5$ per cell. 
 \end{enumerate}
 
The $L_1$ distance and $L_2$ distance between a confidential table $s$ and $\tilde s$, a privatization of $s$, are given respectively by
 \begin{equation}\label{eq:dist}
 L_1\left(s,\tilde s \right)=\sum_{i}\left|s_{i}-{\tilde s}_{i}\right|\quad{\rm and}\quad  L_2\left(s,\tilde s\right)=\sqrt{\sum_{i}\left(s_{i}-{\tilde s}_{i}\right)^{2}}.
 \end{equation}
For each of the four types of privatization and processing mechanisms, we compute averages of both distances over 100 realizations of $\tilde s$. Figure~\ref{fig:mis-variance} displays the box plots of these average distances over the 20 simulated copies of private table $s$. 
We observe that the conditional mechanism, constructed from an unconstrained mechanism with $\epsilon_0 = 0.5$, exhibits a degree of variability between the unconstrained mechanisms with privacy loss budget $\epsilon = 0.5$ and $\epsilon = 1$. On the other hand, the nonnegative $L_2$ projection of the unconstrained mechanism with  $\epsilon = 1$ achieves a level of accuracy mostly on par with it.  Note that this observation should not be taken as a suggestion of relative accuracy between the nonnegative $L_2$ minimization and the constrained mechanism, because the effective privacy guarantee that either mechanism enjoys is undetermined, an issue we will discuss further at the end of Section~\ref{sec:postprocess}. That said, by Theorem~\ref{thm:2epsilon} that the conditional mechanism inflates the privacy loss budget of a unconstrained algorithm  by $\left(1+\gamma\right)$, the empirical observation suggests that the effective $\gamma$ may be somewhere in between $0$ and $1$. In the following sections, we will see examples where $\gamma=0$ or even $\gamma<0$. 

\section{Curator's post-processing may not be innocent processing}\label{sec:postprocess}

A common practice to ensure unconstrained differentially private releases respect the mandated disclosure requirements is through optimization-based post-processing, which takes the general form
\begin{equation}\label{eq:post-process}
\DPost \left(\DP;\InvSet\right)=\text{argmin}_{s\in \InvSet}\Delta\left(\DP,s\right).
\end{equation}
That is, $\DPost$ is the element in $\InvSet$ that is the closest to $M$ according to some discrepancy measure $\Delta$, typically a distance, such as the two given in \eqref{eq:dist}. In the case of the Census Bureau's TopDown algorithm, $\DPost$ is a composite post-processing procedure consisting of first a nonnegative $L_{2}$ minimization followed by an $L_{1}$ minimization onto the integer solutions; see \citep{abowd2019census}. 

In the literature of differential privacy, there is a widely referenced theorem which establishes that differentially privatized releases are ``immune'' to post-processing \citep{dwork2014algorithmic}. The theorem states that if $\DP$ is an $\epsilon$-differentially private mechanism and $g$ is an arbitrary function, then $g \circ \DP$ is still $\epsilon$-differentially private. Indeed for any $g$-measurable set $B$, 
\begin{equation}\label{eq:post-process-theorem}
P\left(g(\DPx)\in B\right)=P(\DPx\in g^{-1}(B)).	
\end{equation}
Thus for every $x \in \mathcal{X}$, the maximal increased risk of disclosure from 
releasing $g(M(x))$ cannot exceed that from releasing $M(x)$ (but the reversed is guaranteed only when $g$ is one-to-one).
Intuitively, further blurring an already blurred picture can only make it harder, not easier, to see what is in the original picture.

This intuition, however, is based on the assumption that the further blurring process does not use any knowledge about the original picture. We need to make clear here that imposing invariants on differentially private releases via optimization-based post-processing, in the sense of the operation discussed here, does {\it not} in general fall under the jurisdiction of the post-processing theorem. This is because $\DPost$, the function used to impose invariants on the unconstrained output $\DPx$, is implicitly dependent on the confidential dataset $x^{*}$, with the dependence induced via $\InvSet$, or equivalently $\mathcal{X}^{*}$ to which $x^{*}$ belongs. Since $\InvSet$ supplies information about the confidential database, whereas the unconstrained mechanism $\DP$ is by design not preferential towards $\InvSet$, any further processing of $\DP$ that  makes nontrivial use of $\InvSet$ risks violating the privacy guarantee that $\DP$ deserves. 

The post-processing theorem guarantees that no loss of privacy will be induced to the privatized query via any functional transformation that may be carried out by an ordinary analyst or data user. However, imposing invariants is the kind of post-processing that only the data curator -- one who has access to the confidential data -- is capable of performing. In the extreme scenario (see Example~\ref{ex:modulus}) that the invariant forces the privatized disclosure to be precisely equal to the confidential query, for the data curator to achieve this algorithmically is as simple as taking the privatized query $\DP$ and projecting it  to the single point in $\mathbb{R}^{d}$ defined by the confidential value $s\left(x^*\right)$. But this is impossible for a data user who do not know what $s\left(x^*\right)$ is. 

One may wonder the following question. While the invariant $\InvSet$ has a dependence on the confidential $x^{*}$, itself is nevertheless public information. Doesn't that make $\DPost(\cdot; \InvSet)$ a fully specified function, just like $g$ in  the post-processing theorem? The answer is no in general, and the distinction here is a subtle one. When talking about the value of the invariants, it suffices to regard $\InvSet$ merely as an announced description of the confidential data. However as alluded to previously, $\InvSet$ is a set-valued map from the database space to subsets of the query space, i.e.  $\InvSet: \mathcal{X} \to \mathscr{B} (\mathbb{R}^{d})$. Almost always is the case in practice that the functional form of map of $\InvSet$ is {\it a priori} determined, but its value -- namely $\InvSet\left(x^{*}\right)$ -- can be calculated only after the confidential data is observed. Indeed, the actual specification of $\InvSet$ would almost certainly change if $x^{*}$ were observed differently. This means for an $f$-measurable set $B$ and a database $x \in \mathcal{X}$, the equivalent events in the $f$ and the $\DP$ spaces are now
\begin{equation}\label{eq:post-process-complication}
f\left(\DPx;\InvSet\left(x\right)\right)\in B\;\Leftrightarrow\;\DPx\in f^{-1}\left(B; \InvSet(x)\right),
\end{equation}
noting that the inverse function $f^{-1}(\cdot;\InvSet(x))$ now depends on $x$.

To see the complication caused by this dependence,  write $\tilde B(x)=f^{-1}(B; {\InvSet(x)})$ and $f(x)=f(\DPx;\InvSet(x))$. We then have 
\begin{align}\label{eq:cross}
\frac{P(f(x)\in B)}{P(f(x')\in B)}  = \frac{P(M(x)\in \tilde B(x))}{P(M(x')\in \tilde B(x'))}.
\end{align}	
Although both $\tilde B(x)$ and $B(x')$ are measurable sets, they are not necessarily the same  when $x'\neq x$. Hence we cannot use \eqref{eq:dp} to conclude that the right hand side of \eqref{eq:cross} is bounded above by $\exp\left(\epsilon\right)$. This does not necessarily imply that the post-processing $f$ as defined in \eqref{eq:post-process} is not $\epsilon$-differentially private. Indeed, we prove in Appendix~A that  both $L_2$ and (a class of) $L_1$ post-processing in Example~\ref{ex:bin} below is $\epsilon$-differentially private.  But it does imply that in general, the statistical and privacy properties of $f$ are not straightforwardly inherited from that of $M$, and hence they need to be established on a case by case basis.

Another major drawback of using optimization-based methods to impose invariants is that the statistical intelligibility of differential privacy is obscured. The post-processing function $f$ is often procedurally defined, hence a complex and confidential data-dependent map from the unconstrained query space to the constrained query space, with almost impenetrable statistical properties, and certainly so for any given database. In contrast, using conditioning to realize differential privacy with mandated disclosure, despite often computationally demanding by construction, preserves the statistical intelligibility of the privacy mechanism. The constrained privacy mechanism is distributionally -- as opposed to procedurally -- constructed, preserving the possibility of transparent downstream analysis. It furthermore delivers privacy guarantee in the same format as does differential privacy without constraints, offering a congenial statistical interpretation that resembles the original.

Below we use an example to compare congenial privacy with two approaches of post-processing for a same query function. The example is simple enough for analytical derivations of the distributions of post-processing mechanisms to be possible. As we will see, the three approaches to impose invariant constraints yield distinct theoretical behaviors.

\begin{example}[A two-bin histogram with constrained total]\label{ex:bin}

Suppose the confidential database $x$ is a binary vector, and the query of interest tabulates the number of $0$ and $1$ entries in $x$, i.e. 
\[
s\left(x\right) = \left(s_{1}\left(x\right),s_{2}\left(x\right)\right) = \left(\sum_i{\bf 1}\left(x_{i}=0\right),\sum_i{\bf 1}\left(x_{i}=1\right)\right).
\]
Employ the Laplace mechanism as the unconstrained privatization mechanism to protect the two-bin histogram, i.e. 
\[
\DP\left(x\right) = \left(m_{1} = s_{1}+u_{1}, m_{2} = s_{2}+u_{2}\right), \; u_{i} \overset{i.i.d.}{\sim} Lap(2\epsilon^{-1}),\] 
expending in total $\epsilon$ privacy loss budget.   The induced probability density of $M$ is
\[
p\left(m_{1},m_{2}\right)=\left(\frac{\epsilon}{4}\right)^{2}\exp\left\{ -\frac{\epsilon}{2}\left(\left|m_{1}-s_{1}\right|+\left|m_{2}-s_{2}\right|\right)\right\} .
\]

Suppose the invariant to be imposed is that the total of the privatized histogram shall be the same as that the the confidential query itself. That is, for any given $x$, the associated invariant set is 
\begin{equation}\label{eq:2bin-invset}
\InvSet\left(x\right) =\left\{ \left(a_{1},a_{2}\right)\in\mathbb{R}^{2}:a_{1}+a_{2}= s_1\left(x\right) + s_2\left(x\right)\right\}.
\end{equation}
In the calculations below, a certain database $x$ is fixed, and we write the invariant total $n = \left\Vert x\right\Vert$, the length of $x$.
\end{example}

\paragraph{Congenial privacy.} Under the constraint of histogram total, our congenial $\DPCo\left(x\right)$ is obtained from the conditional distribution
\[
\left(s_{1}+u_{1},s_{2}+u_{2}\right)\mid u_{1}+u_{2}=0.
\]
It turns out that the probability density of $\DPCo$ is
\begin{equation}\label{eq:laplace-0sum}
P\left(m_1  =  m, m_2 = n-m\right)=\frac{\epsilon}{2}\exp\left\{ -\epsilon\left|m-s_{1}\right|\right\}
\end{equation}
and $0$ otherwise; see Appendix~\ref{appendix:laplace}. That is, our congenial mechanism $M^*$ is simply to draw a $u$ from $Lap(\epsilon^{-1})$, and then release $(m_1=s_1+u,  m_2=s_2-u)$. Clearly, the privacy property of $M^*$ is the same as its first component, call it $M^*_1$, which protects $s_1$ by releasing $m_1$ because setting $m_2=n-m_1$ is a deterministic step with no implication on privacy when $n$ is known. But $M^*_1$ is simply the Laplace mechanism with $\epsilon$ budget. Consequently, our congenial mechanism maintains the same $\epsilon$ guarantee as the original unconstrained mechanism, even though the meaning of protection is different, as we emphasized in Section~\ref{subsec:cond}. 

\paragraph{Post-processing with $L_{2}$ minimization.}  Here we minimize the $L_{2}$ distance between $\DPx$ and the post-processed histogram release, denoted $\DPostii$, subject to its sum being $n$. The solution is
\[
\DPostii\left(\DPx, \InvSet\right)=\text{argmin}_{s\in \InvSet}\left\Vert \DPx-s\right\Vert _{2}{=}\left(\bar{x}+\tilde{u},\bar{x}-\tilde{u}\right),
\]
where $\tilde{u}$ is the average of two independent Laplace random variables with scale $2\epsilon^{-1}$. As can be easily seen  (for example from its characteristic function), $\tilde{u}$ is not a Laplace random variable, but in fact follows distribution
\begin{equation}\label{eq:average-laplace}
\frac{1}{2}Lap(\epsilon^{-1})+\frac{1}{2}SgnGamma(2,\epsilon^{-1}),	
\end{equation}
that is a 50-50 mixture of a Laplace distribution with scale $\epsilon^{-1}$ and a signed Gamma distribution (i.e. a regular Gamma distribution multiplied by a fair random sign) with shape $k=2$ and scale $\epsilon^{-1}$; see Appendix~\ref{appendix:laplace} for derivation. It is worth noting that since a signed Gamma distribution of shape $k = 2$ can be written as the sum of two independent Laplace distributions of the same scale, $\tilde{u}$ is more variable than a single independent Laplace random variable of the same scale. Hence for any $x$, the privatized release using $\DPostii$ will be more variable than that of the congenial privatization $\DPCo$. Intuitively, this suggests that $\DPostii$ should not do worse than $\DPCo$ in terms of privacy protection. Indeed as we will show in Appendix~\ref{appendix:laplace}, $\DPostii$ also achieves the same $\epsilon$-differentially private guarantee.

\begin{table}
\begin{tabular}{rll}
\toprule 
 & $E\left(m_1, m_2\right)$ & $Var\left(m_{1}\right)$\tabularnewline
\midrule
$M^{*}$ & $\left(s_{1}\left(x\right),s_{2}\left(x\right)\right)^{\top}$ & $\frac{2}{\epsilon^{2}}$\tabularnewline
\midrule 
$f_{L_{2}}$ & $\left(\bar{s}\left(x\right),\bar{s}\left(x\right)\right)^{\top}$ & $\frac{4}{\epsilon^{2}}$\tabularnewline
\midrule 
$f_{L_{1}}$ & $\left(s_{1}\left(x\right),s_{2}\left(x\right)\right)^{\top}$ & $ \left[\frac{4}{\epsilon^{2}},\frac{8}{\epsilon^{2}}\right]$\tabularnewline
\bottomrule 
\end{tabular}

\bigskip
\caption{Differentially private two-bin histogram with invariant total: expectation and first-component variance of the conditional ($\DPCo$) and post-processed  ($f_{L_2}$ and $f_{L_1}$) histograms. \label{tab:2bin-comparison}
}	
\end{table}

\paragraph{Post-processing with $L_{1}$ minimization.}
If we change the $L_2$ norm to $L_{1}$ norm in the above, the privatization
mechanism is no longer unique.  There will be infinitely many solutions in the form of
\[
\DPosti\left(\DPx, \InvSet\right)=\text{argmin}_{s\in\InvSet}\left\Vert \DPx-s\right\Vert _{1}:=\left(\tilde{s},n-\tilde{s}\right),
\]
where $\tilde{s}$ only needs to satisfy
\begin{align*}
\tilde{s} & \in\left[\min\left\{s_{1}+u_{1},n-\left(s_{2}+u_{2}\right)\right\},\max\left\{s_{1}+u_{1},n-\left(s_{2}+u_{2}\right)\right\}\right] \\ &\overset{L}{=}\left[s_{1}+\min\left(u_{1},u_{2}\right),s_{1}+\max\left(u_{1},u_{2}\right)\right],
\end{align*}
where $\min\left(u_{1},u_{2}\right)$ and $\max\left(u_{1},u_{2}\right)$ are the minimum and the maximum of two i.i.d. Laplace random variables. In particular, choosing any convex combination of $u_1$ and $u_2$ as the additive noise term to the first entry constitutes a solution, i.e., $\tilde{s}_1=s_{1}+\beta u_{1} + (1-\beta) u_{2}$ for some $\beta \in [0,1]$, and then set $\tilde{s}_2=n-\tilde{s}_1$.  For the rest of the article,  $L_1$ post-processing refers to this convex combination strategy. 

For ease of reference, Table~\ref{tab:2bin-comparison} collects a comparison of the constrained differentially private histogram $\DPCo$ and two the post-processing approaches, $f_{L_2}$ and $f_{L_1}$, in terms of the expectation and variance of the resulting release for a given database $x \in \mathcal{X}$ and confidential query $s$. All expectations are taken with respect to the relevant mechanism. 

We can see that our congenial mechanism has the smallest variance. Because the congenial mechanism and $f_{L_2}$ both carry the same $\epsilon$-privacy guarantee which \textit{cannot} be further improved, we can comfortably declare that $f_{L_2}$ is \textit{inadmissible} because it is dominated by the congenial mechanism,  providing less utility (in terms of statistical precision) without the benefit of increased privacy protection. However, we cannot say that the congenial mechanism dominates $f_{L_1}$ even though it still leads to smaller variance. This is because, as we will prove in Appendix~\ref{appendix:laplace}, the attained  level of privacy guarantee of $f_{L_1}$ is $\epsilon/(2\max\{\beta, 1-\beta\})$, which is never worse than $\epsilon$.  Hence the increased variance under $f_{L_1}$ might be acceptable as a price for gaining more privacy protection. In general, comparing the utility of two privatization mechanisms with the same nominal but different actual privacy loss budget is as thorny an issue as comparing the statistical power of two testing procedures with the same nominal, but different actual, Type I error rates.

\section{Discussion}\label{sec:discussion}

\subsection{Finding better $\gamma$}\label{subset:gamma}

While Theorem~\ref{thm:2epsilon} always holds with $\gamma=1$, it likely sets a  loose bound on the ratio between $P\left(\DPCo\left(x\right)\in B\right)$ and $P\left(\DPCo\left(x'\right)\in B\right)$, hence declaring an overly ``conservative'' nominal level of privacy loss induced by $\DPCo$. Depending on how the invariant $\InvSet$ interacts with the distributional property of the unconstrained mechanism $M$ in a specific instance, $\gamma$ can be shown to take smaller values, adding more ``bang of the buck'' to the privacy loss budget, so to speak.  Three examples are given below.

\begin{example}[trivial invariants]\label{ex:trivial}
Consider the trivial case where the set of invariants does not actually impose any restriction, i.e., $\mathcal{X}^* = \mathcal{X}$. It is then necessarily true that $\InvSet=\mathcal{S}$, and the ``constrained'' differentially private mechanism is identical in distribution to the unconstrained one: $M^{*} \overset{L}{=} M$. In this case, $\gamma = 0$ and $M^{*}$ is $\epsilon$-differentially private.
\end{example}


\begin{example}[rounding and secrecy]\label{ex:modulus}	
Let $x$ be a binary vector of length $n$ indicating a group of individuals' possession of a certain feature (yes 1/no 0), and the query of interest is $s(x) = \lceil \sum x_i/10 \rceil$, or the number of groups of size $10$ that can be formed by people who possesses the feature. A Double Geometric mechanism $\DP(x) = s(x) + U$ is used to protect the query, with a privacy loss budget of $\epsilon$ (under the global sensitivity of $\nabla(s) = 1$). 

Suppose the following invariant set is mandated for disclosure:
\begin{equation*}
\mathcal{X}^{*}=\left\{ \left(x_{1},\ldots,x_{n}\right)\in\left\{ 0,1\right\} ^{n}:\sum x_{i}\in\left[41,50\right]\right\},	
\end{equation*}
or equivalently, $
\InvSet=\left\{ 5\right\}$
is the singleton set that contains nothing but the true value $s(x^*) = 5$. In this case, the implied constrained privacy mechanism $\DPCo$ is equivalent to a degenerate distribution: $P\left(M^{*}\left(x\right)=5\right)=1$ for all $x \in \mathcal{X}^{*}$. Furthermore, for all neighboring datasets $\left(x,x'\right) \in \mathcal{X}^{*} \times \mathcal{X}^{*}$, and any $B$ a measurable subset of $\mathbb{N}$,  
\[
P\left(\DPCo\left(x\right)\in B\right)=\exp\left(0\right) P\left(\DPCo\left(x'\right)\in B\right)=\begin{cases}
\begin{array}{c}
1\\
0
\end{array} & \begin{array}{c}
\text{if }5 \in B\\
\text{otherwise}.
\end{array}\end{cases}
\]
Therefore in this particular instance, $\DPCo$ is in fact $0$-differentially private, corresponding to $\gamma = -1$ in Theorem~\ref{thm:2epsilon}. This means that for  those databases conformal to the invariant $\mathcal{X}^{*}$, $\DPCo$ supplies no discriminatory information among them whatsoever.  Indeed, if the value of the supposedly private query is already public knowledge, no mechanism can further increase its disclosure risk, therefore achieving complete differential privacy.
\end{example}
 

Our third example is a less trivial example of $\gamma<0$, which is actually provided by the congenial mechanism in Example~\ref{ex:bin}. There, although the guaranteed privacy loss budget is still $\epsilon$, in applying Theorem~\ref{thm:2epsilon}, $k$ must be set to $2$ or greater, because under the constraint of fixed sum, the nearest neighbors among binary vectors $(x, x')$ must have $\mathtt{d}(x, x')=2$.  Hence the $\epsilon$ privacy bound implies $k(1+\gamma)=2(1+\gamma)=1$, yielding $\gamma=-0.5$.

This example also reminds us that a major cause of information leakage due to invariants is the structural erosion to the underlying data space, such as making $\mathtt{d}(x, x') = 1$ (as measured on the original space $\mathcal{X}$) impossible.  In reality, the unconstrained data space $\mathcal{X}$ is typically regular, and contains $\mathcal{X}^*$ as a proper subset. We should expect to find many $x \in  \mathcal{X}^*$, and many (if not many more) $x' \in \mathcal{X} \backslash \mathcal{X}^*$ such that $x$ and $x'$ are  neighbors, near or far. Knowing that the confidential dataset must belong to $\mathcal{X}^*$ categorically rules out the possibility that all the $x'$'s can be the confidential dataset, weakening the differential privacy promise by eliminating the neighbors. If the invariant is sufficiently restrictive such that $\mathcal{X}^*$ becomes topologically small relative to $\mathcal{X}$, it may be the case that for some $x \in \mathcal{X}^*$, all of its original immediate neighboring datasets are not in $\mathcal{X}^*$:
\begin{equation*}
\left\{ \left(x,x'\right)\in\mathcal{X}^{*}\times\mathcal{X}^{*}:\mathtt{d}\left(x,x'\right) = 1 \right\}  = \emptyset,
\end{equation*}
in which case we say that the neighboring structure of the original data space of the database is {\it substantially disrupted},  as seen in Example~\ref{ex:bin}. If the disruption is so substantial that neighbors of any distance cease to exist, we say that the neighborhood structure is completely {\it destroyed}:
\begin{equation*}
\left\{ \left(x,x'\right)\in\mathcal{X}^{*}\times\mathcal{X}^{*}:\mathtt{d}\left(x,x'\right) \ge 1 \right\}  = \emptyset.	
\end{equation*}
Then, even for the constrained mechanism $\DPCo$, the $\epsilon$-differential privacy promise becomes vacuously true, since no possible neighboring pairs remain for which the concept of privacy is applicable. However, Example~\ref{ex:modulus} demonstrates that vacuous privacy promise can occur without the neighborhood structures completely destroyed. 

In general, it is conceptually difficult to parse out the share of responsibility on privacy attributable to the data curator under any scenario of mandated disclosure. If certain information is made public, then any information that it logically implies cannot be expected to be protected, either. The best that we can expect any privacy mechanism to deliver, then, is protection over information that truly remains. Notions that serve the equivalent purposes as $\mathcal{X}^{*}$ and $\InvSet$ have been proposed in the literature for expositions of new notions of differential privacy, including {\it blowfish} and {\it pufferfish} privacy \citep{he2014blowfish,kifer2012rigorous}, where the privacy guarantee is re-conceptualized on the restricted space {\it modulo} any structural erosion to the original sample space due to external or auxiliary information available to an attacker. When interpreting the promise of Theorem~\ref{thm:2epsilon}, we shall pay due diligence to the case in which immediate neighbors no longer exists, and talk about the $\epsilon$-differential privacy guarantee only for those $k$-neighbors that actually do.




\subsection{Other interpretations of privacy}

The literature has seen other lines of work that offer interpretations of differential privacy in statistical terms. Notably, the {\it posterior-to-posterior semantics} of differential privacy \citep{dinur2003revealing,kasiviswanathan2014semantics,ashmead2019effective} explains the effect of privacy also in the vocabulary of Bayesian posteriors. The posterior-to-posterior semantics establishes differential privacy as a bound for the ratio of posterior probabilities assigned to an individual confidential data entry, when the private mechanism is applied to neighboring datasets that differ in only one entry. The said ratio is between the two quantities 
\begin{equation}\label{eq:pair1}
\pi \left(x_{i}^{*}=\omega\mid M_{\epsilon}\left(x\right)\in B\right) \; \text{and}\; \pi \left(x_{i}^{*}=\omega\mid M_{\epsilon}\left(x'\right)\in B\right),	
\end{equation}
where $x$ and $x'$ are neighboring datasets. What varies between the two posterior quantities is the confidential dataset on which the private query is applied. The datasets $x$ and $x'$ are neighboring datasets, one of which presumably (but not necessarily) contains the true value of the $i$th confidential data entry $x^*_{i}$, and the other contains a fabricated value of it. 

The comparison in \eqref{eq:pair1} raises the question of what it means by the conditional probability of $x^*_{i}$ given a private query constructed from a database that does {\it not} contain this true value, as this conditional probability hinges on  external knowledge about how a fabricated database may inform the actual confidential database.  Our \textit{prior-to-posterior semantics} formulated in Theorem~\ref{thm:posterior} takes a practical point of view and avoids such conceptual complication. We compare the disclosure risk before and after an \textit{actual} release, reflecting the core idea behind differential privacy.

\subsection{Full privacy or vacuous knowledge}

As alluded to in Section~\ref{subsec:intelligibility}, the notion of vacuous knowledge cannot be appropriately captured by ordinary probabilities. The defect reflects a fundamental inability of the language of probability in expressing a true lack of knowledge, a central struggle in the Bayesian literature that motivated endeavors in search for the so-called ``objective priors''  \citep{ghosh2011objective}. Neither the uniform distribution nor any other reference distributions are truly information-free, as they all invariably invoke some principle of indifference in relation to a specific criterion (such as the Lebesgue measure, the counting measure, or the likelihood function) which is subject to debate. 
 
To supply the rigorous definition needed to define probabilistically  
$M_0(x)$ in Section~\ref{subsec:intelligibility}, we invoke the concept of lower probability functions, and as a special case belief functions \citep[see e.g.][]{gong2019simultaneous}, both generalized versions of a probability function which amounts to a set of probability distributions on a given space. The statistical information contained in $M_{0}$ is represented by the vacuous lower probability function, denoted $\underline{P}$, which takes the value $\underline{P}(B) = 1$ only when $B = \mathbb{R}^{d}$, and $0$ everywhere else. Equivalently stated in terms of its conjugate upper probability function $\overline{P}(B) = 1 - \underline{P}(B^c)$,
\begin{equation}
\overline{P}\left(B\right)=\begin{cases}
\begin{array}{c}
1\\
0
\end{array} & \begin{array}{c}
\text{if }B\in\mathscr{B}\left(\mathbb{R}^{d}\right)\backslash\left\{ \emptyset\right\} ,\\
\text{if }B=\emptyset.
\end{array}\end{cases}
\end{equation}
That is, the statistical information contained in $M_{0}$ can be (but is not known to be) concordant with {\it any} Borel-measurable probability function, thus the probability of any $B$ is as low as $0$ and as high as $1$, as long as $B$ is neither the full set nor the empty set.

Generally, the conditioning operation involving lower probability functions is not trivial, and it is not unique due to the existence of several applicable rules. But if the lower probability functions being conditioned on is vacuous, there is consensus among different rules as to what posterior distribution should result, namely precisely as stated in \eqref{eq:s0x}. See \cite{gong2017judicious} for an extended exposition of conditioning rules involving lower probability and belief functions.
 
\subsection{Future directions}

This work points to several future directions of pursuit. On the computational front, how do we construct efficient algorithms to realize congenial privacy, by drawing possibly high-dimensional releases subject to complex constraints? When we use non-perfect Markov chain Monte Carlo to accomplish this task, how do we ensure the declared privacy guarantee is not destroyed because a chain cannot run indefinitely?  On the privacy front, for every constrained mechanism constructed through conditioning, how to find the best $\gamma$ value that tracks as closely as possible the effective privacy loss budget, which in turn enables fair performance comparisons among invariant-respecting algorithms? Furthermore, how to achieve an orthogonal decomposition of the public, invariant information from the free, residual information that remains in the confidential microdata, in a logical sense without having to resort to the probabilistic vocabulary of statistical independence?

\begin{acks}
The authors thank Jeremy Seeman, Salil Vadhan, Guanyang Wang and three anonymous reviewers for  helpful suggestions. Ruobin Gong gratefully acknowledges research support by the National Science Foundation (NSF) DMS-1916002. Xiao-Li Meng also thanks NSF for partial financial support while completing this article. 
\end{acks}


\bibliographystyle{ACM-Reference-Format}
\bibliography{master}

\newpage
\appendix

\section{Privacy Guarantees for Congenial, $L_1$ and $L_2$ methods}\label{appendix:laplace}

We first derive the conditional distribution of two i.i.d. Laplace random variables, given that their sum is zero.  Let $u_{1},u_{2}\overset{i.i.d.}{\sim}Lap\left(2\epsilon^{-1}\right)$
and denote $v=u_{1}$, $w=u_{1}+u_{2}$. Since $\left(v,w\right)$
is linear in $\left(u_{1},u_{2}\right)$, their joint probability density function is given by.
\begin{eqnarray*}
p\left(v,w\right) & \propto & p\left(u_{1}\left(v,w\right),u_{2}\left(v,w\right)\right)\\
 & \propto & \exp\left(-0.5\epsilon\left|v\right|-0.5\epsilon\left|w-v\right|\right),
\end{eqnarray*}
This implies that 
\begin{eqnarray*}
p\left(v\mid w=0\right)& \propto & p\left(v,w=0\right) \\&\propto&\exp\left(-\epsilon\left|v\right|\right)\sim Lap(\epsilon^{-1}),	\end{eqnarray*}
which leads to \eqref{eq:laplace-0sum}.

We then derive the density for $\tilde u= \beta u_1 + (1-\beta) u_2$, where $\beta\in (0, 1)$ (the case of $\beta=0$ or $1$ is trivial). This  covers the $L_1$ projection case, where any $\beta\in [0, 1]$ is acceptable, and the $L_2$ projection case, where $\beta=1/2$.   Since $u_1=\beta^{-1}[\tilde u - (1-\beta) u_2]$, the Jacobian from $(\tilde u, u_2)$ to $(u_1,  u_2)$ is $\beta^{-1}$. Consequently, 
\[
p_\epsilon\left(\tilde{u},u_{2}\right)=\frac{\epsilon^{2}}{16\beta}\exp\left\{ -\frac{\epsilon}{2\beta}\left[\left|\tilde{u} - (1-\beta)u_{2}\right|+\beta\left|u_{2}\right|\right]\right\}.
\]
To derive $p_\epsilon(\tilde u)$, we assume without loss of generality $\tilde u \ge 0$. Consider $p_\epsilon\left(\tilde{u}\right) = \int p_\epsilon\left(\tilde{u},u_{2}\right)du_{2}$ on three regions: 
\begin{align*}
I_1(\beta)
&= \frac{\epsilon^{2}}{16\beta} \int_{-\infty}^0\exp\left\{ -\frac{\epsilon}{2\beta}(\tilde{u} - u_{2})\right\}d u_2
= \frac{\epsilon}{8} \exp\left\{ -\frac{\epsilon}{2\beta}\tilde{u}\right\}; \\
I_2(\beta) &= \frac{\epsilon^{2}}{16\beta} \int_0^{\frac{\tilde u}{1-\beta}} \exp\left\{ -\frac{\epsilon}{2\beta}\left[\tilde{u}+ (2\beta-1)u_{2}\right]\right\}d u_2\\
&=\frac{\epsilon}{8(2\beta-1)}\left[\exp\left\{ -\frac{\epsilon}{2\beta}\tilde{u}\right\} -  \exp\left\{- \frac{\epsilon}{2(1-\beta)}\tilde{u}\right\}\right]; \\
I_3(\beta) & = \frac{\epsilon^{2}}{16\beta} \int_{\frac{\tilde u}{1-\beta}}^\infty \exp\left\{ -\frac{\epsilon}{2\beta}\left[u_2- \tilde{u}\right]\right\}d u_2 
=\frac{\epsilon}{8} \exp\left\{ -\frac{\epsilon}{2(1-\beta)}\tilde{u}\right\}.
\end{align*}
Summing up these terms and noting the symmetry of 
$p_\epsilon$, we obtain
\begin{align}\nonumber 
p_\epsilon(\tilde u) & =\frac{\epsilon}{4(2\beta-1)} \left[\beta\exp\left\{ -\frac{\epsilon}{2\beta}|\tilde{u}|\right\} - (1\!-\!\beta)\exp\left\{ -\frac{\epsilon}{2(1\!-\!\beta)}|\tilde{u}|\right\}\right]\\
& =  \frac{\beta^2}{(2\beta\!-\!1)}Lap\left(2\beta\epsilon^{-1}\right) - \frac{(1\!-\!\beta)^2}{(2\beta\!-\!1)}Lap\left(2(1\!-\!\beta)\epsilon^{-1}\right).
\label{eq:nega}
\end{align}

\medskip
\paragraph{Remark I} 
The expression \eqref{eq:nega} is fascinating. It shows that the density of a convex combination of i.i.d. Laplace random variables is a ``mixture" but non-convex combination of two Laplace densities with different scale parameters, because 
\[
\frac{\beta^2}{(2\beta-1)}+ \left[- \frac{(1-\beta)^2}{(2\beta-1)}\right] = 1.
\] 
That is, although the two weights add up to one, they always take the opposite sign when $\beta\not=1/2.$

\paragraph{Remark II} 
When $\beta=1/2$, the expression \eqref{eq:nega} is of $0/0$ appearance, but is well-defined once taking the limit $\beta\rightarrow 1/2$ and using the L'Hospital's rule, yielding
\begin{equation}
p_\epsilon(\tilde u) = \frac{\epsilon}{4}\left[1+\epsilon |\tilde u|\right]\exp\left\{-\epsilon |\tilde u|\right\}.\label{eq:bhalf}  
\end{equation}
This can be written as 
\begin{eqnarray*}
p\left(\tilde{u}\right)
 & = & \frac{1}{2}\cdot\frac{\epsilon}{2}\exp\left\{ -\epsilon\left|\tilde{u}\right|\right\} +\frac{1}{2}\cdot\frac{\epsilon^{2}}{2}\left|\tilde{u}\right|\exp\left\{ -\epsilon\left|\tilde{u}\right|\right\}\nonumber \\
 & \sim & \frac{1}{2}Lap\left(\epsilon^{-1}\right)+\frac{1}{2}SgnGamma\left(2,\epsilon^{-1}\right),
\end{eqnarray*}
suggesting that it is more variable that $Lap(\epsilon^{-1})$. We now prove that \eqref{eq:bhalf}, moreover the entire family of distributions given by \eqref{eq:nega} as indexed by $\beta\in (0, 1)$, is $\epsilon$-differentially private. In fact, they attain a level of privacy protection more stringent than $\epsilon$  whenever $\beta\not=1/2$. Our proof relies on the following general result, which can be useful for verifying  differential privacy guarantees in other situations. 

\begin{theorem} Suppose $f(x)$ is a positive real-valued function on a normed vector space $\mathcal{X}$, with its norm denoted by $|x|$. Suppose $f(x)$ has the following properties: 
\begin{description}
\item[(i)] $f(x)$ is monotone decreasing in $|x|$;
\item[(ii)] $g_\alpha(x)=f(x)e^{\alpha|x|}$ is monotone increasing in $|x|$, where $\alpha$ is a positive constant.
\end{description}  
Then for any $a\in \mathcal{X}$ and $b\in \mathcal{X}$, we have 
\begin{equation}\label{eq:gbound}
    \sup_{x\in \mathcal{X}} \frac{f(x-a)}{f(x-b)}\le e^{\alpha|a-b|}.
\end{equation} \label{thm:key}
\end{theorem} 
\begin{proof}  For any $x\in \mathcal{X}$, if $|x-a|>|x-b|$, then $f(x-a)\le f(x-b)$ by (i) and hence \eqref{eq:gbound} holds trivially.   If $|x-a|\le |x-b|$, then $g_\alpha(x-a)\le g_\alpha(x-b)$ by (ii), and hence 
\begin{align*}
    \frac{f(x-a)}{f(x-b)} & =\frac{g_\alpha(|x-a|)}{g_\alpha(|x-b|)}e^{\alpha(|x-b|-|x-a|)}\\
    & \le e^{\alpha(|x-b|-|x-a|)}\le e^{\alpha|a-b|}.
\end{align*}
\end{proof}
To apply this result, we first note that 
for $p_\epsilon(x)$ of \eqref{eq:nega} with $x>0$, its derivative is given by
\begin{align}\nonumber 
\frac{d p_\epsilon(x)}{dx} & =\frac{\epsilon^2}{8(2\beta-1)} \left[ \exp\left\{ -\frac{\epsilon}{2(1\!-\!\beta)}x\right\}- \exp\left\{ -\frac{\epsilon}{2\beta}x\right\}\right]<0,
\end{align}
for any $\beta\not=1/2$. For $\beta=1/2$, we can directly verify from \eqref{eq:bhalf} that
\begin{align}\nonumber 
\frac{d p_\epsilon(x)}{dx} & =-\frac{\epsilon^3}{4} x\exp\left\{ -\epsilon x\right\}<0,
\end{align}
Hence, condition (i) holds for \eqref{eq:nega} for all $\beta\in (0, 1)$. 

To establish condition (ii), the choice $\alpha$ is
the key since it directly governs the degree of privacy guarantee. From the expression \eqref{eq:gbound}, we want the smallest $\alpha$ such that condition (ii) holds, which gives us the tightest bound hence better privacy guarantee. A good strategy here is to start with $\alpha=c\epsilon$ and let the mathematics tell us how to minimize over $c$. We start with the simplest case with $\beta=1/2$. From the expression \eqref{eq:bhalf}, the smallest $c$ that can make $g_\alpha$ monotone increasing is clearly 1. The resulting $g_\alpha$ also has the  property that
\begin{align}\label{eq:limt}
    \lim_{|x|\rightarrow \infty}  \frac{g_\alpha(|x-a|)}{g_\alpha(|x-b|)}=1.
\end{align}
This implies that the bound $e^{\epsilon |a-b|}$ can be approached arbitrarily closely by letting $|x| \rightarrow\infty$, which means that the privacy loss budget $\epsilon$ cannot be reduced.  For our current application, this means the post-processing by $L_2$ projection is also differentially private at level $\epsilon$, but not more stringent than that.

When $\beta\not= 1/2$, we assume without loss of generality $\beta>1/2$. Then for $g_\alpha(x)=e^{c\epsilon |x|} p_\epsilon(x)$, it is easy to verify that for any $x>0$,
\begin{align}\nonumber 
\frac{d g_\alpha(x)}{dx}\!=\!\frac{\epsilon^2}{8(2\beta-1)} 
\!\left[w_c\exp\left\{ \frac{w_c}{2\beta}\epsilon x\right\}\!+\![w_c+c']\exp\left\{-\!\frac{w_c+c'}{2(1-\beta)}\epsilon x\right\}\right],
\end{align}
where $w_c=2c\beta -1$ and $c'=2(1-c)$. Our job is to seek the smallest $c$ such that this derivative  is non-negative regardless of the value of $x$. Clearly the positivity holds when we set $w_c=0$, that is $c=(2\beta)^{-1}<1$, and hence $c'>0$. To show that this is the smallest possible $c$,  we see that when setting $\alpha=\epsilon/(2\beta)$, we have 
\begin{align*}
     \frac{g_\alpha(|x-a|)}{g_\alpha(|x-b|)}=
     \frac{\beta- (1\!-\!\beta)\exp\left\{ -\tau |x-a|\right\}}{\beta- (1\!-\!\beta)\exp\left\{ -\tau |x-b|\right\}},
\end{align*}
where $\tau =(2\beta-1)/(2\beta(1\!-\!\beta))>0$. Clearly as $|x| \rightarrow \infty$, the ratio above goes to $1$ regardless of the value $a$ and $b$ as long as they are fixed. Consequently, the same implication from \eqref{eq:limt} follows, that the bound $e^{\alpha|a-b|}$ can be approached arbitrarily closely with $\alpha=\epsilon/(2\beta)$, hence it cannot be further improved. That is, for post-processing via $L_1$ projection, the actual differential privacy protection achieved is $\epsilon/(2\beta)$ when $\beta>1/2$  (and  $\epsilon/(2(1-\beta))$ when $\beta<1/2$).  This makes intuitive sense. For example, when $\beta=1$ the injected noise is drawn from a single $Lap(2\epsilon^{-1})$ distribution, corresponding to a privacy loss budget of $\epsilon/2$.

In summary, for any $\beta\in [0, 1]$, the attained privacy loss budget for  $M(x)=s(x)+\tilde u$ is $\epsilon/(2\max\{\beta, 1-\beta\})$.

\section{Sampling scheme for the Double Geometric distribution}

The Double Geometric mechanism, as introduced in Definition~\ref{def:2geo}, utilizes additive noise whose cumulative mass function is given by
\begin{equation*}
F\left(u\right)=P\left(U\le u\right)=\begin{cases}
\begin{array}{c}
\frac{a^{-u}}{1+a}\\
1-\frac{a^{u+1}}{1+a}
\end{array} & \begin{array}{c}
u\le0,\\
u>0,
\end{array}\end{cases}	
\end{equation*}
with quantile function
\begin{equation*}
F^{-1}\left(v\right)=\begin{cases}
\begin{array}{c}
\left\lceil \frac{-\log v-\log\left(1+a\right)}{\log\left(a\right)}\right\rceil \\
\left\lfloor \frac{\log\left(1-v\right)+\log\left(1+a\right)}{\log\left(a\right)}\right\rfloor 
\end{array} & \begin{array}{c}
v\le\frac{1}{1+a},\\
v>\frac{1}{1+a}.
\end{array}\end{cases}	
\end{equation*}
Hence, one way to sample a Double geometric random variable is via inverse probability sampling. That is,
\begin{equation*}
	 U_{i}\sim Unif\left(0,1\right), \qquad F^{-1}\left(U_{i}\right)\sim {p_i\left(\cdot \mid \epsilon\right)},
\end{equation*}
where $p_i$ is given by \eqref{eq:2geo}. This method is implemented for all numerical examples illustrated in this paper.

\section{Performance Diagnostics of the MIS Algorithm}\label{appendix:mis}

The acceptance rate of Algorithm~\ref{alg:MIS} rate is shown in  Figure~\ref{fig:mis-dispersion} as a function of  the proposal inverse scale parameter $\tilde{\epsilon}$ . The acceptance rate is the highest in this example when $\tilde{\epsilon}$ is set to $0.6$, just slightly larger than the privacy loss budget of the unconstrained privacy mechanism ($\epsilon = 0.5$ per cell). The acceptance rate achieved is about $1.68\%$.

Figure~\ref{fig:mis-trace} shows traceplots of $10,000$ draws from Algorithm~\ref{alg:MIS} of respectively the second (in proposal index set $\mathcal{I}$) and the first (not in proposal index set $\mathcal{I}$) cells of the constrained differentially private contingency table, when $\tilde{\epsilon} = 0.6$.
 
\begin{figure}
   \centering
   \includegraphics[width=.4\textwidth]{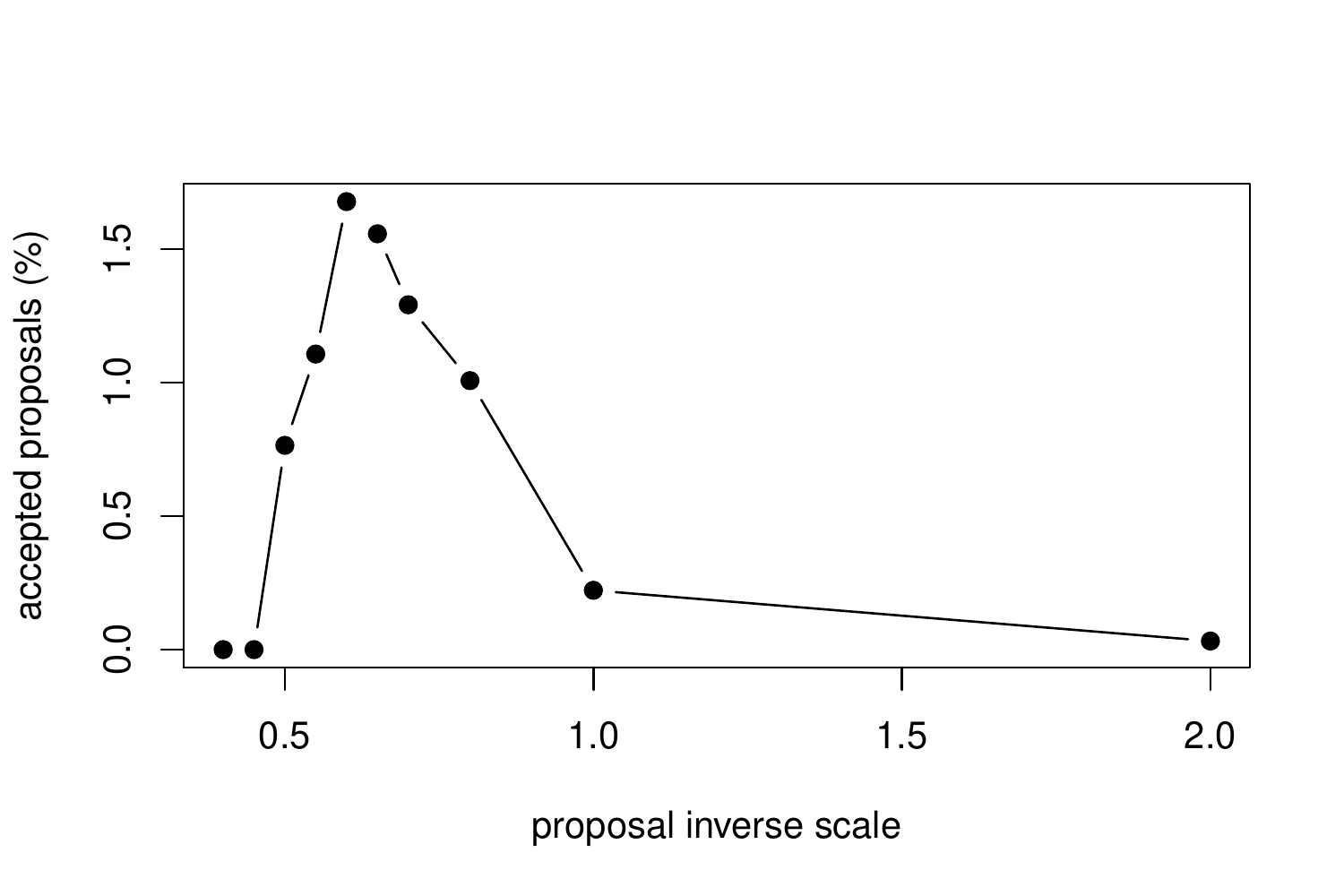}
   \caption{Algorithm~\ref{alg:MIS} acceptance rate as a function of the proposal parameter $\tilde{\epsilon}$. \label{fig:mis-dispersion}}
\end{figure}

 \begin{figure}
   \centering
   \includegraphics[width=.5\textwidth]{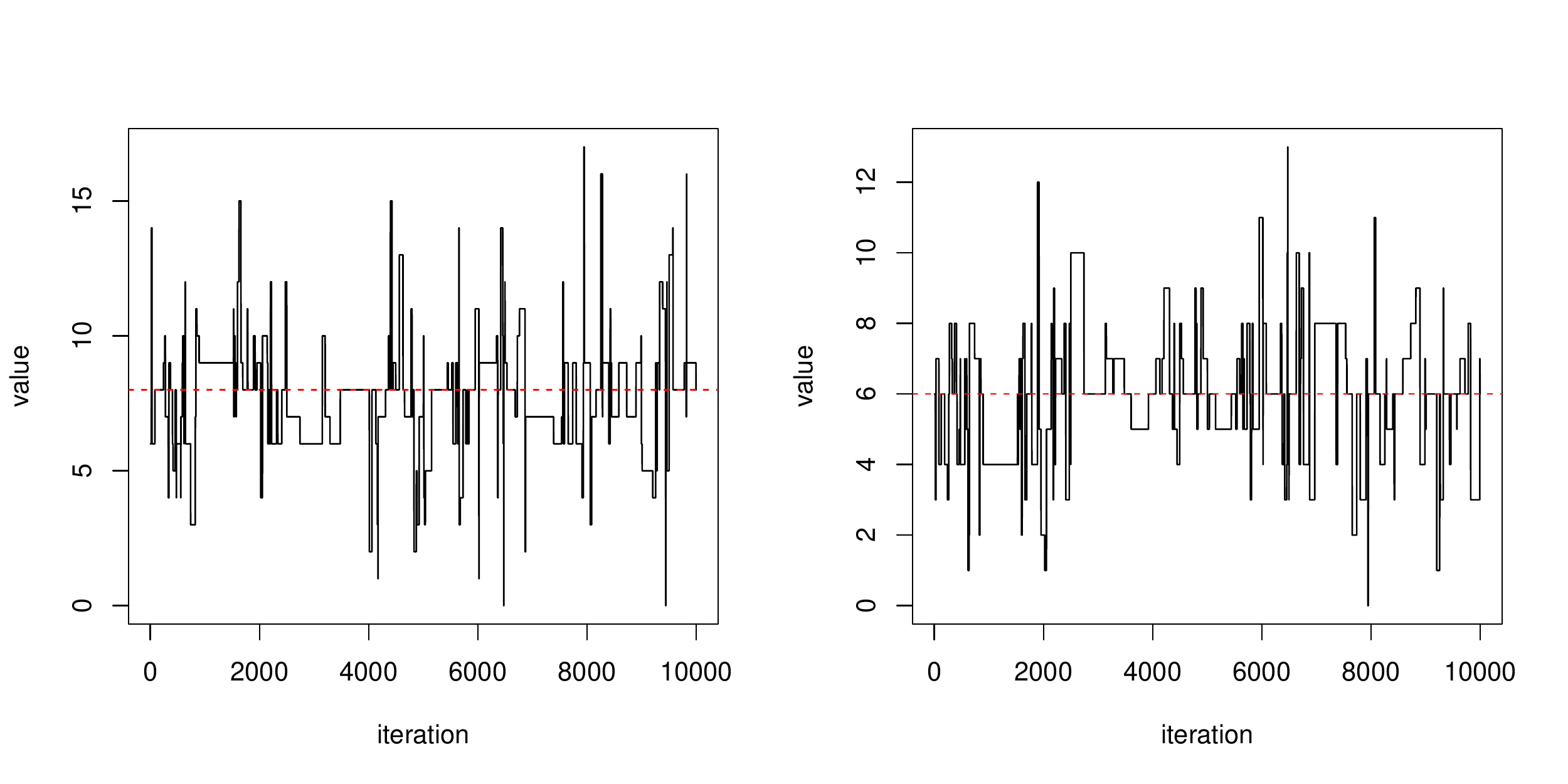}
   \caption{Traceplots of $10,000$ draws from Algorithm~\ref{alg:MIS} of the second (left) and the first (right) cell of the constrained differentially private contingency table.\label{fig:mis-trace}}
 \end{figure}

\end{document}